\documentclass[english,11pt,a4paper]{amsart}
\usepackage{amstext,amsthm,amsfonts,amsmath,amssymb}
\usepackage{graphicx}
\usepackage{babel}
\usepackage{color} 
\usepackage{hyperref}

\makeatletter 
\newcommand{\B}{\mathcal{B}}
\newcommand{\D}{\mathcal{D}}

\newcommand{\A}{\mathcal{A}}
\newcommand{\K}{\mathcal{K}}

\renewcommand{\H}{\mathcal{H}}

\newcommand{\N}{\mathbb{N}}

\newcommand{\R}{\mathbb{R}}
\newcommand{\C}{\mathbb{C}}

\newcommand{\unit}{\mathbf{1}}
\newcommand{\ra}{\rightarrow}

\newcommand{\tr}{\operatorname{tr }}
\newcommand{\Ad}{\operatorname{Ad }}

\newcommand{\rmd}{\operatorname{d\!}}
\newcommand{\rmi}{\operatorname{i}}
\newcommand{\rme}{\operatorname{e}}

\renewcommand{\d}{\delta}
\newcommand{\eps}{\varepsilon}

\newcommand{\id}{\operatorname{id}}

\def\d{\mathrm{d}}

\newtheorem{theorem}{Theorem}[section]
\newtheorem{corollary}[theorem]{Corollary}
\theoremstyle{definition}\newtheorem{definition}[theorem]{Definition}
\theoremstyle{definition}\newtheorem{option}{Option}
\theoremstyle{definition}\newtheorem{remark}[theorem]{Remark}
\theoremstyle{definition}
\theoremstyle{definition}
\theoremstyle{definition}

\title{Control of Quantum Noise: on the role of dilations}

\author{Daniel Burgarth}
\address[DB]{Center for Engineered Quantum Systems, Department of Physics and Astronomy, Macquarie University, Sydney, 2109 NSW, Australia}
\author{Paolo Facchi}
\address[PF]{Dipartimento di Fisica and MECENAS, Università di Bari, I-70126 Bari, Italy, and INFN, Sezione di Bari, I-70126 Bari, Italy}
\author[Robin Hillier]{Robin Hillier$^*$}
\address[RH]{Department of Mathematics and Statistics, Lancaster University, Lancaster LA1 4YF, UK}
\email{r.hillier@lancaster.ac.uk}
\thanks{$^*$ corresponding author}
\date{25/7/2022}

\begin{document}
\begin{abstract}
We show that every finite-dimensional quantum system with Markovian (i.e., GKLS-generated) time evolution has an autonomous unitary dilation which can be dynamically decoupled. Since there is also always an autonomous unitary dilation which cannot be dynamically decoupled, this highlights the role of dilations in the control of quantum noise. We construct our dilation via a time-dependent version of Stinespring in combination with Howland's clock Hamiltonian and certain point-localised states, which may be regarded as a C*-algebraic analogue of improper bra-ket position eigenstates and which are hence of independent mathematical and physical interest.  
\end{abstract}
\maketitle

\section{Introduction}
The rich structure of quantum noise has inspired generations of mathematicians, engineers  and physicists alike~\cite{Accardi,Alicki, Attal1, Attal2,Attal3, Bengtsson, Breuer, Davies, Exnerbook, Gardiner, Gough, Jacobs, LB13, Nurdin, Rivas, Wal, Weiss, Wiseman}. Roughly speaking one can distinguish between a microscopic and a macroscopic approach: while the former considers all degrees of freedom, the latter looks at effective averaged dynamics. The bridge between these two approaches is formed by a rich variety of exact and approximate reduction methods (`micro to macro': Born-Markov approximation, Weak and Strong Coupling Limits, Partial Traces and Conditional Expectations, Projection Operator Techniques, Path Integral Approaches, Hierarchical Equations of Motion) and by exact and approximate extension  methods (`macro to micro': Quantum Stochastic Calculus, Gregoratti-Chebotarev Hamiltonian, Collision Models). 

The key motivation of the reduction approach is easy to understand. On a foundational and mathematical level it helps us to understand how decay, decoherence and the thermodynamical laws arise.  From an engineering perspective, one obtains models which are much easier to solve. The motivation to extend is, however, more subtle. Why would one try to find complicated extensions  (henceforth referred to as `dilations') of simple models which give the same dynamics anyway? And why would one not simply take the \emph{actual} physical model which is found in the lab?

One motivation is foundational and mathematical: to understand the mathematical structure obtained in the reduced model, one inverts the process to find dilations. This is a beautiful realisation theory, in which one aims to characterise non-uniqueness and how qualitative features of the reduced theory are reflected in their dilations. Dilations help understanding if the reduction approach has led to physically acceptable models, as exemplified by Stinespring's theory of complete positivity~\cite{Stinespring55,Paulsen2003}. The second, more practical motivation comes from engineering and control: since the true actual physical dilation is often unknown, the mathematical dilation can act as a \emph{surrogate} model for the physics, and opens new pathways of interfacing quantum noise. This approach has led to the important concepts of quantum filtering and feedback control, which interfaces not only with the macroscopic but with the microscopic description \cite{Wiseman,Nurdin}. In quantum optics such surrogate dilations appear to mimic the true physics very well, and their applications play a key role in the emerging field of quantum engineering, where they can help to fight quantum noise (the key obstacle to quantum computation).
 
The motivation of the present work is to understand how good such an  approach is \emph{beyond} the specific field of quantum optics. For instance, some of the leading quantum technology is based on superconducting physics, in which much richer noise models come to play~\cite{White}. Specifically, we consider one of the key control methods for noisy systems---dynamical decoupling. This open loop control is robust, as it requires very little knowledge of the true model and works for a large class of models. We defer the specifics of dynamical decoupling until later, but emphasise that its robustness makes it a good candidate for benchmarking  the response quantum noise gives to control through the underlying microscopic picture~\cite{AHFB15,Burgarth19,GN, Pollock, WisemanM}.
 
Since till date quantum stochastic calculus is the only known exact and generic dilation giving rise to a time-independent Hamiltonian (Gregoratti-Chebotarev~\cite{GoughG}), we first introduce a new, second dilation method, which allows us to compare the two approaches. This method is based on a combination of a time-dependent Stinespring theory~\cite{DMB} and a clock system extension~\cite{How}. The clock system makes use of point-localised quantum states which are modelled as singular states on bounded operators on a Hilbert space  or C*-subalgebras. They may be regarded as an algebraic relative of improper eigenstates of position and momentum operators arising in the Dirac bra-ket language. However,  unlike the latter, they are normalised and thus are valid quantum-mechanical states with a clear probabilistic interpretation. Since we could not find a treatment in literature, we decided to develop them here from scratch, which makes this part of independent physical and mathematical interest.
 
Then in Theorem \ref{th:decoupling} and Corollary \ref{cor2}, we compare the control performance of dynamical decoupling on these two dilations. For Gregoratti-Chebotarev, Gough and Nurdin~\cite{GN} have already shown recently that dynamical decoupling does not work. We will show that for our dilation, the opposite is true. This shows that the surrogate approach to control of noisy systems depends strongly on the underlying microscopic physics. Since the latter is typically unknown, it appears difficult to predict whether dynamical decoupling will work for a given macroscopic model. Both types of dilations are idealisations, and the physical truth is generally not fully captured by either of them but they offer at least good approximations and rigorous mathematical frameworks. It is however worth noting that dynamical decoupling has been successfully used in numerous solid systems and NMR-based experimental realisations of quantum computer systems, cf.~\cite{NISQ22} (in particular pages 27-28 and references there) for a recent survey article, which indicates that the Gregoratti-Chebotarev 
dilation might be an unrealistic description for those physical systems.

Our paper is organised as follows. In Section~\ref{sec:phys} we provide a qualitative and rather conceptual explanation of our dilation for which dynamical decoupling works.  We illustrate this with a simple single-qubit amplitude damping model. This section is in particular intended for readers who do not want to deal with the detailed mathematical aspects of this work. Section~\ref{sec:singular} provides a number of constructions of point-localised states and discusses their mathematical properties. These are used in the main Section~\ref{sec:DD}, where we recall the basics of dynamical decoupling and dilations, and then construct the aforementioned dilation for which we prove that dynamical decoupling works. We conclude with a brief outlook.

\section{Physical summary and single-qubit illustration}\label{sec:phys}

Commonly, in open quantum systems, one starts with a time-independent Hamiltonian on a large system, and finds effective noisy dynamics on a subsystem. Intuitively it is clear that the inverse of this procedure, albeit highly non-unique, should also be possible. This is what is meant by a \emph{dilation}. Here, we construct such dilations by a series of steps summarised in Figure~\ref{scheme}. 
\begin{figure}[tbp!]
\includegraphics[height=0.6\textwidth]{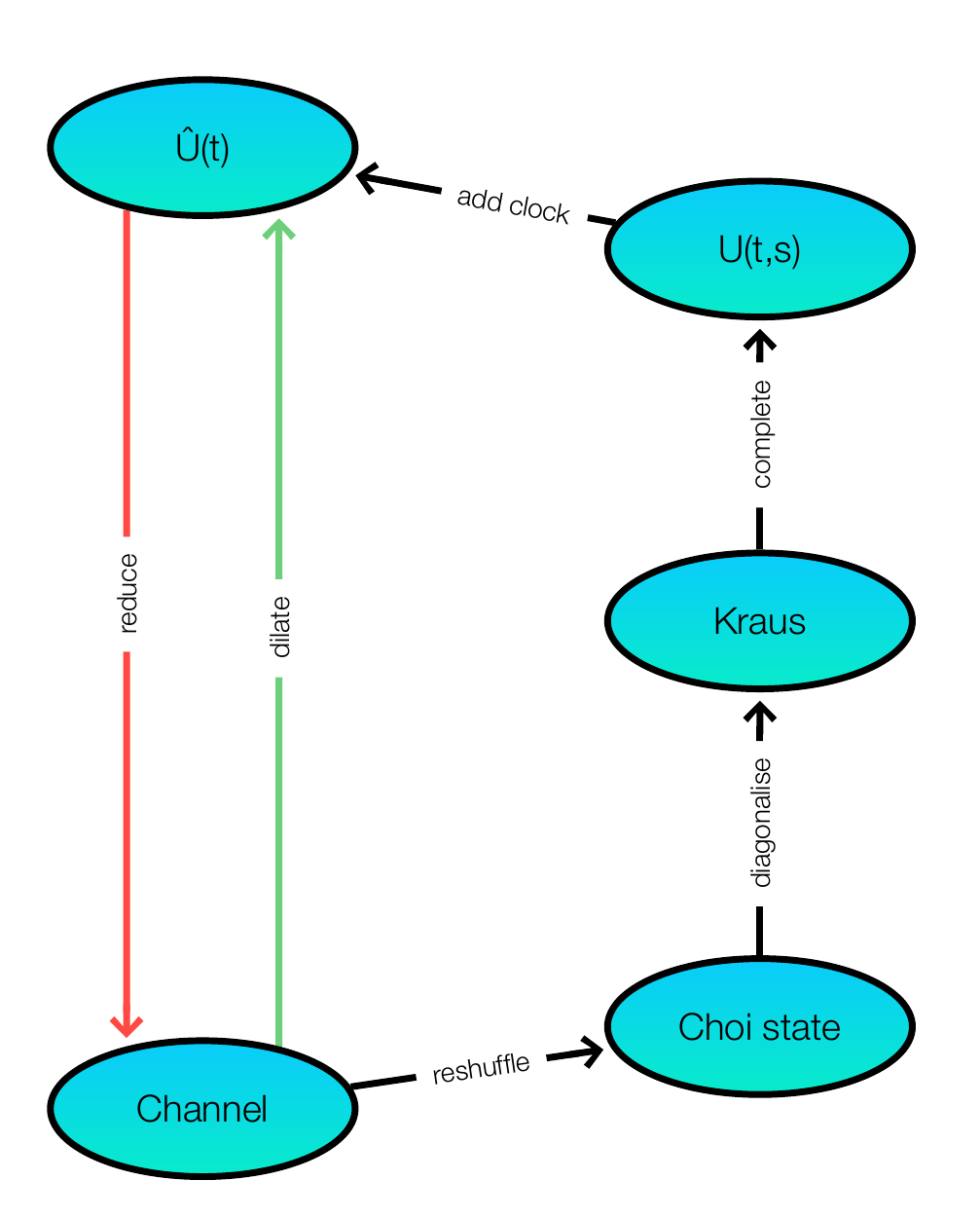}
\caption{\label{scheme} Summary of the scheme.  For a time-dependent channel such as $\rme^{t \mathcal{L}}$, we `dilate' (green arrow) by constructing a unitary group $\hat{U}(t)=\rme^{-\rmi t\hat{H}}$ and an initial state on a larger system such that the partial trace (red arrow) over the extended dynamics recovers the channel perfectly. This dilation is constructed via a number of steps (black arrows), extending~\cite{DMB} by the last step of adding a clock.}
\end{figure}

It is best to start with an example. Consider the \emph{single-qubit amplitude damping} Lindbladian on a two-level system $\C^2$, 
cf.~\cite{AHFB15},
\begin{equation}\label{eq:ad}
	\mathcal{L}(\rho)=\sigma_-\rho\sigma_+-\frac{1}{2}
(\sigma_+\sigma_-\rho+\rho\sigma_+\sigma_-),
\end{equation} 
where $\sigma_+ = \begin{pmatrix} 0 & 1 \\ 0 & 0 \end{pmatrix}$ and $\sigma_- = \sigma_+^*$ are Pauli raising/lowering operators acting on the qubit density matrix $\rho$  ($\rho \geq 0$, $\tr \rho = 1$). This gives rise to a channel, for $t\geq 0$,
\begin{equation*}
	\rme^{t\mathcal{L}}
	\begin{pmatrix}
\rho_{00} & \rho_{01} \\
\rho_{10} & \rho_{11} 
\end{pmatrix}
=\begin{pmatrix}
\rme^{-t} \rho_{00}  & \rme^{-\frac{t}{2}} \rho_{01} \\
\rme^{-\frac{t}{2}} \rho_{10} &  (1-\rme^{-t}) \rho_{00} + \rho_{11}
\end{pmatrix},
\end{equation*} 
which corresponds to a purely exponential decay to the pure density matrix $\rho_\infty=|1\rangle\langle 1|$.
It is easy to see that such a decay can also come from a flip-flop interaction with a second qubit initialised in the pure density matrix $\mu =|1\rangle\langle 1|$: this interaction preserves excitations, so the initial density matrix $\rho=|1\rangle\langle 1|$ is a steady state; while the initial density matrix $|0\rangle\langle 0|$ decays into  $|1\rangle\langle 1|$ initially, corresponding to a transfer of the flip to the second system. Of course, if we want to mimic decay, we need the coupling between the two qubits to eventually go down to zero. More precisely, the following choice on $\C^2\otimes \C^2$ does the job~\cite{AHFB15}:
\begin{equation}\label{eq:H12}
	H_{12}(t)= g(t) H_{1\leftrightarrow 2}, \qquad H_{1\leftrightarrow 2} = \sigma_+\otimes\sigma_-+\sigma_-\otimes \sigma_+,
	\qquad g(t) = \frac{1}{2\sqrt{e^{|t|}-1}}, \quad t\neq0.	
\end{equation}
This time-dependence does not only ensure that the coupling goes to zero in a way that the excitation never returns to the first qubit; it also becomes strong enough for short times such that the initial decay is already exponential. For $t\rightarrow 0$ the Hamiltonian diverges, but in an integrable way, so that the time evolution operator $U(t)=\exp\bigl(- \rmi  G(t) H_{1\leftrightarrow 2}\bigr)$, with $G(t)=\int_0^t g(s) \d s$, is well defined and continuous for all times, with $U(0)=1$. 

In the present manuscript, we want to extend this Hamiltonian further, to a larger bath space, such that it becomes time-independent. Informally, this corresponds to adding a  \emph{clock system} which automatically switches on the interaction and tunes it. Such clock systems are commonly used in quantum thermodynamics (see for instance~\cite{Malabarba2014}). In mathematical physics, they have been introduced by Howland~\cite{How} in a rigorous way. Roughly, the idea is to think as time represented by the position $x$ of a particle moving on a line, and to consider the time-independent Hamiltonian 
\begin{equation}\label{howlanddil}
	\hat{H}=-\rmi\frac{\d}{\d x}\otimes \unit \otimes \unit + g(x) \otimes H_{1\leftrightarrow 2},
	\end{equation}
in which the previous time-dependent coupling has become an unbounded potential term $g(x)$ for the clock particle, which moves by rigid translations with constant speed, generated by the momentum operator $P =-\rmi \frac{\d}{\d x}$.
In Section~\ref{sec:DD} we will make this formal expression rigorous by showing that indeed such a self-adjoint time-independent dilation can be defined. 
We also show how such a dilation can be found for any Lindbladian channel $e^{t\mathcal{L}}$ and more generally for analytic families of completely positive trace-preserving (CPTP) maps.

It should be expected that the clock particle only provides an exact dilation if its  initial  state is a point-localised state at  $x=0$. One might expect that this is only possible with a wave function which is a non-physical delta function (that is, not a valid normalised state). Such objects are often informally introduced as `improper Dirac-vectors' or more rigorously  as elements of a rigged Hilbert space. However, it turns out that such objects are too singular for our purpose. Instead we construct a state in an algebraic picture that fulfills our dilation purposes and that should be regarded as the correct point-localised state of the bath. Hand-wavingly this might be considered as the 'square root of a delta function'. Interestingly enough, it turns out that the clock can be chosen to be classical, as its observables are commutative.

Once we have found a dilation of a noisy dynamics, we may ask whether such a dynamics can be dynamically decoupled, that is, can the system-bath interaction be averaged away by quick rotations of the system? In specific cases, this is known to be a highly dilation-dependent question~\cite{GN, ABFH}. A numerical simulation for the single-qubit amplitude damping model~\eqref{eq:ad} shown in  Figure~\ref{simulation}  \begin{figure}[tbp!]
\includegraphics[width=0.5\textwidth]{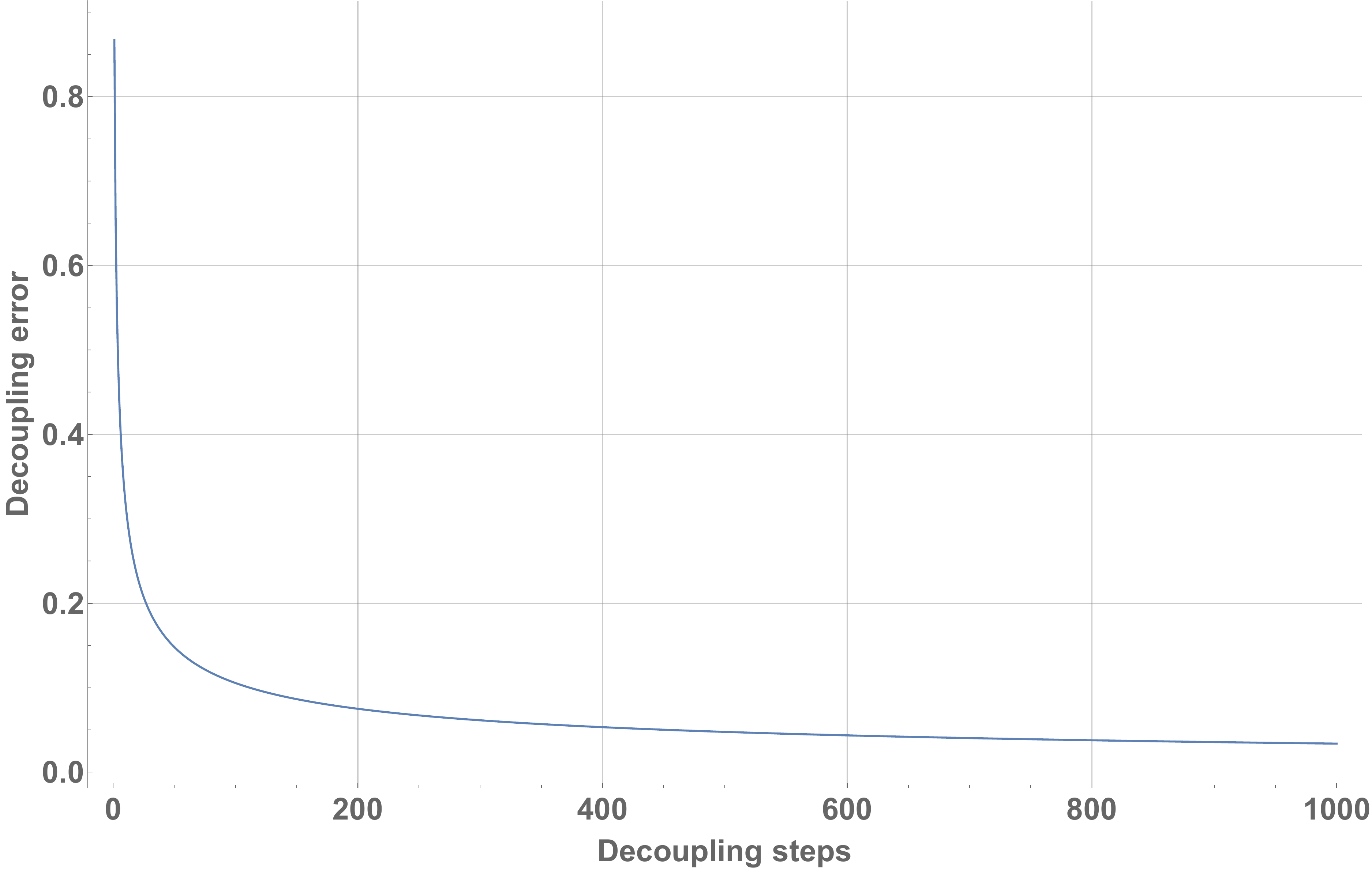}
\caption{\label{simulation} Numerical simulation of the decoupling error for the Hamiltonian Eq.~(\ref{howlanddil}) dilating the single-qubit amplitude damping model. We show an upper bound to the decoupling error (defined later in Eq.~(\ref{eq:toprove})) as a function of the pulse time steps $t/\Delta t$ for a total time $t=1$.}
\end{figure}
shows that the decoupling error approaches $0$ as the decoupling time steps $\Delta t$ tend to $0$, i.e. as the number of decoupling operations $t/\Delta t$ tends to $\infty$. In Section~\ref{sec:DD} we provide a mathematical proof of this fact, not only for the amplitude damping model but for any reasonable finite-dimensional noisy quantum system. For the sake of simplicity, we only deal with bang-bang dynamical decoupling operations~\cite{VL98}. See, however, Remark \ref{remarkeuler}. Roughly, the reason why decoupling works is that the clock state does not change the dynamics, and that the singularity at $t=0$ in~\eqref{eq:H12} integrates to something continuous. This is suprising, because previous studies~\cite{GN,JEK,Class} seemed to indicate that dilations of amplitude damping might be particularly bity.  

\section{Point localisation and singular states}\label{sec:singular}

Consider the separable Hilbert space $\K=L^2(\R)$ and a function $\psi=\unit_{[0,1]}$, which has $\|\psi\|_2 =1$. 
Then for every $t\in\R$ and $\eps>0$, let
\[
\psi_{t,\eps} = \frac{1}{\sqrt{\eps}} \psi\Big( \frac{\cdot- t}{\eps} \Big),
\]
which has support in $[t,t+\eps]$ and again $\|\psi_{t,\eps}\|_2 =1$.
Let
\[
\nu_t (A) = \lim_{\eps\downarrow 0} \langle \psi_{t,\eps}, A \psi_{t,\eps} \rangle,
\]
for suitable $A\in B(\K)$; we would like to understand which $A$ are suitable. To this end, we offer four options, explain their properties and compare them.

\begin{option}\label{sub:1}
Let $X$ and $P$ denote the usual position and momentum operators on $\K$ on their natural domains $\D(X)=\{\xi\in\K : X\xi\in\K\}$ and $\D(P)=\{\xi\in\K : P\xi\in\K\}=H^1(\R)$, respectively. We define
\[
\A= \{f(X):\; f\in C_b(\R)\}\simeq C_b(\R),
\]
where $C_b(\R)$ is the algebra of bounded continuous functions with sup norm $\|\cdot\|_\infty$, an Abelian C*-subalgebra of $B(\K)$. Then $\nu_t$ is well-defined on $\A$ and in fact $\nu_t = \delta_t$, the $\delta$-distribution centred at $t$, or in other words, the (continuous) evaluation functional $f(X) \in \A \mapsto f(t)\in\C$. This is actually a pure state.
Notice the translation property:
\begin{equation*}
\nu_t(T(s)^* A T(s)  ) = \nu_{t+s} (A),
\end{equation*}
for all $A\in \A$, where $s \mapsto T(s)= \rme^{-\rmi  s P}$ is the group of translations ($T(s) \varphi(x) = \varphi(x-s)$, $\varphi \in \K$).

Using~\cite[Prop.2.3.24]{BR1} and~\cite[Sect.3]{La}, one obtains:

\begin{theorem}\label{th:Option1}
$\nu_t$ is a well-defined state on $\A$. It may be extended to a pure state $\hat{\nu}_t$ on $B(\K)$. This extension is singular (i.e. it cannot be associated with a density operator) and possibly not unique.
\end{theorem}

\end{option}

\begin{option}\label{sub:1a}
Let
\[
\tilde{\A} = C^*\{f(X), g(P):\; f,g \in C_\infty(\R)\}
\]
be the C*-algebra generated by functions $f$ and $g$ of the position and momentum operators, 
where $C_\infty(\R)$ is the algebra of bounded continuous functions (w.r.t. the sup norm $\|\cdot\|_\infty$) such that their limits at $\pm\infty$ exist, .

\begin{theorem}\label{th:Option1a}
$\nu_t$ is a well-defined state on $\tilde{\A}$. Moreover, $[f(X),g(P)]\in K(\K)$, and $\nu_t$ vanishes on all compact operators.
\end{theorem}

We postpone the proof after the next theorem.

The algebra $\tilde{\A}$ and slightly bigger algebras, e.g. the ones generated by functions with vanishing oscillations at $\pm\infty$ are extensively studied in harmonic analysis and pseudo-differential operators, cf.~\cite[Sec.3]{Cor} and applications of the limit state construction.

\end{option}

\begin{option}\label{sub:2}
We define
\[
\D=\{A\in B(\K) : \;(\forall t\in\R)  \lim_{\eps\downarrow 0} \langle \psi_{t,\eps}, A \psi_{t,\eps} \rangle \textrm{ exists} \}.
\]
\begin{theorem}\label{th:Option2}
The following properties hold:
\begin{itemize}
\item[(i)] $\D$ is a *-closed Banach space;
\item[(ii)] $\D$ is not an algebra, in particular $\D\subsetneq B(\K)$;
\item[(iii)] $\D$ contains the unital algebra $\tilde{\A}$;
\item[(iv)] $\D$ is $\sigma$-weakly dense in $B(\K)$;
\item[(v)] $\D$ is translation-invariant.
\item[(vi)] $\nu_t$ is a singular state on $\D$, i.e., a unital positive linear map $\D\ra\C$ which is not $\sigma$-weakly continuous. It extends to a singular state $\bar{\nu}_t$ on $B(\K)$.
\end{itemize}
\end{theorem}

\begin{proof}
Let us write $\nu_{t,\eps}= \langle \psi_{t,\eps}, \cdot \psi_{t,\eps} \rangle$.

(i) It follows from the linearity of $\nu_{t,\eps}$ and $\lim_{\eps\downarrow 0}$ that $\D$ is a vector space. Moreover, for $A\in\D$, we have
\[
\lim_{\eps\downarrow 0}\nu_{t,\eps} (A^*) =\lim_{\eps\downarrow 0}\overline{\nu_{t,\eps} (A)} =\overline{\nu_t (A)},
\]
so $A^*\in\D$, too. Finally, $\nu_t$ is norm-continuous, namely we have
\[
|\nu_t(A)| = \lim_{\eps\downarrow 0} |\nu_{t,\eps}(A)| \le \|A\|.
\]
To see that $\D$ is norm-closed, consider a Cauchy sequence $(A_n)_{n\in\N}$ in $\D$ with limit $A\in B(\K)$. Since $\nu_t$ is bounded by $1$, $(\nu_t(A_n))_{n\in\N}$ forms a Cauchy sequence in $\C$ and hence converges to some limit $\ell\in\C$. Given $\lambda>0$, we can now find $n_0\in\N$ such that $\|A_n-A\|<\lambda/3$ for all $n>n_0$. Then there is $\eps_0>0$ such that for all $0<\eps<\eps_0$
\[
\|\nu_{t,\eps}(A) - \ell \| \le \|\nu_{t,\eps}(A) - \nu_{t,\eps}(A_n) \| + \|\nu_{t,\eps}(A_n) - \nu_t(A_n) \|+ \|\nu_t(A_n) - \ell \|
< \frac{\lambda}{3} + \frac{\lambda}{3}  + \frac{\lambda}{3} = \lambda,
\]
which proves that $A\in\D$ and $\nu_t(A) = \ell$, so $\D$ is norm-closed, thus a Banach space.

(ii) Consider the function
\[
f:\R\ra\R, \quad f = 1+ \sum_{n\in\N} (-1)^{n-1} \unit_{[2^{-n},2^{-n+1}]}.
\]
This is bounded and hence $f(X)\in B(\K)$. We compute
\[
\nu_{0,2^{-n}}(f) = 2^n \int_0^{2^{-n}} f(t) \rmd t
\]
and find that the sequence  $(\nu_{0,2^{-n}}(f(X)))_{n\in\N}$ alternates between $\nu_{0,1}(f(X))=\frac53$ and $2-\nu_{0,1}(f(X)) =\frac13$, so 
$(\nu_{0,2^{-n}}(f(X)))_{n\in\N}$ does not converge, hence $f(X)\not\in\D$. Now $f$ is positive, so let $g$ be the positive square-root of $f$ as a real-valued function. Let $A=T(2)g(X)$, the composition of the left shift by $2$ (a unitary operator in $B(\K)$) and multiplication by $g$. Then $\nu_{t,\eps}(A)=0$, for all $t\in\R$ and all $0<\eps<1$. Hence $A,A^*\in\D$ and $\nu_t(A)=0=\nu_t(A^*)$ whereas $A^*A = f(X)\not\in\D$. So $\D$ is not an algebra.

(iii) Given $f\in C_\infty(\R)$, we get
\[
|\langle \psi_{t,\eps}, f\rangle| = \|f\|_\infty \frac{1}{\sqrt{\eps}} \int_0^{\eps} \psi(s/\eps)\rmd s 
= \|f\|_\infty \sqrt{\eps} \int_0^1 \psi(s)\rmd s  \ra 0, \quad \eps\downarrow 0.
\]
Hence the rank-1 operator $|f\rangle\langle f |$ lies in $\D$. If instead $f\in\K$ then it can be approximated in $\|\cdot\|_2$ by a sequence $(f_n)_{n\in\N}$ in $C_\infty(\R)$, so for any $\lambda>0$ we find $n\in\N$ large enough and hence $\eps>0$ small enough such that 
\[
|\langle \psi_{0,\eps}, f\rangle | \le |\langle \psi_{0,\eps}, f-f_n\rangle |  + |\langle \psi_{0,\eps}, f_n\rangle | \le \|f-f_n\|_2 +   |\langle \psi_{0,\eps}, f_n\rangle | < 2 \frac{\lambda}{2} = \lambda.
\]
Thus the rank-1 operator $|f\rangle\langle f |$ lies in $\D$ and in fact $\nu_t(|f\rangle\langle f |)=0$. Since $\D$ is a Banach space, we get that the closure of the span of rank-1 operators $K(\K)$ lies in $\D$ and $\nu_t$ vanishes on $K(\K)$.

The statement about functions of the position operator $X$ is clear from Option~\ref{sub:1} above. Concerning the momentum operator $P$, we use Fourier transformation and find
\begin{align*}
\nu_{t,\eps}(f(X)g(P)) = \langle \psi_{t,\eps}, f(X)g(P) \psi_{t,\eps} \rangle 
= \langle f(X) \psi_{t,\eps}, g(P) \psi_{t,\eps} \rangle 
= \langle f(X) \psi_{t,\eps}, (g(X) \hat{\psi}_{t,\eps})^\vee \rangle ,
\end{align*}
where $\hat{\psi}$ denotes the Fourier transform and $\psi^\vee$ denotes the inverse Fourier transform of $\psi$ in $\K$; the latter equation follows from the spectral theorem for continuous functions. From the structure of $\psi_{t,\eps}$, it follows that
\[
\langle \psi_{t,\eps}, f(X)g(P) \psi_{t,\eps} \rangle  
\ra f(0)\frac{1}{2} \Big(\lim_{p\ra \infty}g(p) + \lim_{p\ra -\infty}g(p) \Big), \quad \eps\downarrow 0,
\]
provided the limits $\lim_{x\ra \pm\infty}f(x)$ exist -- this is the only condition $f$ has to satisfy.

According to~\cite[Lem.3.9.1]{Cor}, an arbitrary monomial in functions of $X$ and $P$ may be written as $f(X)g(P) + k$ where $k\in K(\K)$, and according to the previous steps, this shows that the monomial lies in $\tilde{\A}$. By continuity this extends to the norm-closure $\tilde{\A}$ as in the proof of part (i) above.

(iv) It is well-known that $K(\K)$ is $\sigma$-weakly dense in $B(\K)$, so by (iii) $\D$ is $\sigma$-weakly dense in $B(\K)$.

(v) This is obvious from the construction: if $A\in\D$ then $\nu_{t,\eps} (T(s)^* A T(s)) = \nu_{s+t,\eps} (A)$, so  $T(s)^* A T(s)$ (translation by $s$) lies in $\D$.

(vi) To see unitality and positivity, notice that $\nu_{t,\eps}$ are unital and positive, hence the limit $\nu_t$ is unital and positive as well. Hence $\nu_t$ may be regarded as a state on a $\sigma$-weakly dense subspace of $B(\K)$. Since $\nu_t(A) = 0$ for every $A\in K(\K)$ but $\nu_t(\unit)=1$ and $\unit$ can be approximated by $K(\K)$ in the $\sigma$-weak topology, we see that $\nu_t$ is not continuous in that topology and such a state is called singular~\cite[Def.III.2.15]{Tak}.

Since $\D\subset B(\K)$ is a norm-closed subspace, $\nu_t$ can be extended to a linear functional $\bar{\nu}_t$ on $B(\K)$ with $\|\bar{\nu}_t\|=\bar{\nu}_t(\unit)=1$, due to the Hahn-Banach extension theorem. It follows from~\cite[Prop.2.3.11]{BR1} that $\bar{\nu}_t$ is positive, hence a state on $B(\K)$.
\end{proof}

\begin{proof}[Proof of Theorem~\ref{th:Option1a}]
This follows from (iii) in the preceding proof, restricting $\nu_t$ from $\D$ to $\tilde{\A}$. 
\end{proof}
\end{option}

\begin{option}\label{sub:3}

Consider the nonseparable Hilbert space $\K=L^2(\R,\mu)$ with $\mu$ being the counting measure and with standard scalar product and canonical orthonormal basis $(e_t)_{t\in\R}$. Then on all of $B(\K)$, for every $t\in\R$, we define the normal pure (vector) state
\[
\nu_t : B(\K)\ra\C,\quad \nu_t(A) = \langle e_t, A e_t\rangle . 
\]

Notice that time evolution in this case is not strongly continuous and hence Stone's theorem cannot be applied. For the same reason, the state $\nu_t$ cannot be obtained through a limiting procedure as in the previous options.
\end{option}

\subsection*{Discussion and comparison}

The physical meaning of $\nu_t$ is that of a state which is point-localised at $t$. One often works with idealised point-localised states in physics when dealing with wave functions. These wave functions do not lie in $\K$ but instead in a larger space $\mathcal{X}$ including distributions and they pair with test function space $\mathcal{X}^*$. The inclusion $\mathcal{X}^*\subset\K\subset\mathcal{X}$ is called a Gelfand triple and the corresponding theory is described in the framework of rigged Hilbert spaces. While rigged Hilbert spaces have a fairly elaborated theory, in particular the decomposition into an improper eigenbasis, they  do not offer a description of the corresponding states on $B(\K)$. Our  Option~\ref{sub:2} is an attempt to provide such a description, namely to construct a state on $B(\K)$ describing an idealised point-localisation. As in the case of wave functions, this state cannot pair with the whole algebra $B(\K)$ but only with a (dense) subspace $\D$, which forms the counterpart of $\mathcal{X}^*$ above. Since $B(\K)$ is not a Hilbert space in a natural form, we cannot speak of a generalised eigenbasis decomposition and we do not have the nice structure here that Gelfand triples offer. However, Theorem~\ref{th:Option2} explains the most important properties of $\D$. It is the biggest possible space where all $\nu_t$ can be defined, and this definition is physically motivated and meaningful.

The state $\nu_t$ in Options~\ref{sub:1},~\ref{sub:1a} and~\ref{sub:2} can be extended to states on $B(\K)$, namely $\hat{\nu}_t$ and $\bar{\nu}_t$. It is not clear to us whether these three extensions may coincide. In other words, it is unknown whether the extensions are unique and whether they coincide on $\D$. It is also unclear what the physical meaning of these extensions could be.

A completely different approach is Option~\ref{sub:3}. In this case we have got a normal (i.e. associated to a density operator) and pure state but at the expense of separability and strong continuity of the time evolution, which makes it non-standard in quantum mechanics.

From a mathematical point of view it might be more convenient to work with a nice C*-algebra. For this reason we have provided Option~\ref{sub:1}, namely $\A\subset\D$, which will be sufficient for our dynamical decoupling application in the following section and in addition is Abelian, hence in a sense a classical system. As seen above, $\nu_t$ acts nontrivially on many elements in $\D\setminus\A$, so there are ``quantum effects" in $\D$ which are lost when working with $\A$. The reasoning in the following section works equally well for Options~\ref{sub:1a},~\ref{sub:2} and~\ref{sub:3}.

\section{The dilation and dynamical decoupling}\label{sec:DD}

Our setting is the following. Consider a finite-dimensional quantum system with Hilbert space $\H_0$ of finite dimension $d_0$ and an \emph{analytic time-dependent quantum channel}, that is a continuous one-parameter family of unital completely positive (UCP) maps 
\[
\Phi_t: B(\H_0) \to B(\H_0), \qquad t\ge 0,
\] 
which is analytic as a function of $t\geq0$ (i.e.\ it is the restriction to $[0,\infty)$ of a holomorphic function), 
and  with $\Phi_0=\id.$ An important example of such a family is given by a \emph{quantum dynamical semigroup} $(\Phi_t=\rme^{t \mathcal{L}})_{t\ge 0}$~\cite{Davies}, whose generator $\mathcal{L}$ has the Gorini-Kossakowski-Lindblad-Sudarshan (GKLS) form~\cite{GKS,Lindblad}.

Since we will have to work with non-normal states, for which there are no corresponding density matrices, we will refrain from using density matrices in this section and instead write everything in terms of (algebraic) states. The time-evolution maps act then naturally on the algebra of observables, i.e., we use the Heisenberg picture.

 We would  like to construct a C*-algebraic autonomous unitary dilation to a larger system in the following sense:
\begin{definition}\label{def:unidilation}
Let $\{\Phi_t\}_{t\geq 0}$ be a one-parameter family of UCP maps on a C*-algebra $\A \subset B(\H_0)$. A \emph{unitary dilation of} $(\A,\Phi_t)$ is defined as a triple $\bigl(\B\otimes\A,  \omega_e, \hat{U}(t)\bigr)$, where $\B \subset B(\H_e)$ is a C*-algebra on a Hilbert space $\H_e$,  $\omega_e$ is a state on $\B$, and $\{\hat{U}(t)\}_{t\in\R}$ is a family of unitary operators on the product $\hat{\H}=\H_e \otimes \H_0$, with the property that
\[
(\omega_e  \otimes \varrho )\circ \Ad (\hat{U}(t)^*)(\unit\otimes A) = \varrho \circ\Phi_t(A),
\]
for all $A\in \A$ and all states $\varrho$ on $\A$, where $\Ad (\hat{U}(t)^*) = \hat{U}(t)^* \cdot \hat{U}(t)$.   We say that the unitary dilation is \emph{autonomous}  if $\hat{U}(t)$ is a one-parameter group. We say that the unitary dilation is \emph{strongly-continuous} if the map $t\in\R\mapsto \hat{U}(t)\xi\in\hat{\H}$ is continuous, for every $\xi\in\hat{\H}$.
\end{definition}

Notice that $\A$ may be regarded as the algebra of observables of the quantum system; when it is the full algebra of bounded operators on a Hilbert space, that is $\B\otimes\A= B(\H_e) \otimes B(\H_0)= B(\hat{\H})$, and $\omega_e$ is a normal state associated with a density matrix, Definition~\ref{def:unidilation} reduces to the standard definition of unitary dilation of a system in $\H_0$ with an environment in $\H_e$, commonly used in quantum information. Here, in order to accomodate a quantum clock with  point-localised states, we extend the definition by allowing singular states defined on sub-algebras. According to the options discussed in the previous section we could equally extend such singular states to the whole $B(\hat{\H})$ without affecting our results.

We remark that a unitary dilation is a particular  C*-algebraic dilation $\bigl(\B\otimes\A, \Ad (\hat{U}(t)^*), \iota, P\bigr)$ in the notation of~\cite[Def.2.3]{SS}, where 
\[
\iota (A) = \unit\otimes A, \quad \varrho \circ P(\hat{A}) = (\omega_e\otimes\varrho ) (\hat{A}), \quad \forall A\in \A, \; \hat{A}\in \B\otimes \A.
\]
This is just meant to be reference and we will not need this general framework of C*-algebraic dilations in the present paper.

The second ingredient we need is a decoupling scheme for the given system. We recall~\cite[Def.2.1]{ABFH}:

\begin{definition}\label{def:decset}
A \emph{decoupling set for $\H_0$} is a finite set of unitary operators $V\subset B(\H_0)$ such that
\[
\frac{1}{|V|} \sum_{v\in V} v A v^* = \frac{1}{d_0} \tr(A) \unit,\quad \forall A\in B(\H_0).
\]
If $n$ is a multiple of the cardinality $|V|$ then a \emph{decoupling cycle} of length $n$ is a cycle $(v_1, v_2,\ldots,v_n)$  through $V$ which reaches each element of $V$ the same number of times. 
\end{definition}

Henceforth we fix $V$.

\begin{definition}
With the above notation and a given autonomous unitary dilation $(\B\otimes B(\H_0),  \omega_e, \hat{U}(t))$, the time evolution under the influence of a decoupling cycle $(v_1, v_2,\ldots,v_n)$ is given by
\begin{equation}\label{eq:def-Unt}
\hat{U}^{(n)}(t)= \hat{v}_n \hat{U}\Bigl(\frac{t_n-t_{n-1}}{n}\Bigr) \hat{v}_n^* \cdots \hat{v}_2\hat{U}\Bigl(\frac{t_2-t_1}{n}\Bigr)\hat{v}_2^* \hat{v}_1\hat{U}\Bigl(\frac{t_1-0}{n}\Bigr) \hat{v}_1^*,
\end{equation} 
where $\hat{v}_j = \unit \otimes v_j$, $j=1,\dots,n$. 

We say that \emph{dynamical decoupling works} if, for every given $t>0$ and $\eps>0$, there is $n\in\N$ such that for every decoupling cycle of length $n$ in $V$ and every initial state $\varrho$ of the system, we have
\begin{equation}\label{eq:toprove}
E_n(\varrho)= \| (\omega_e\otimes\varrho )\circ \Ad(\hat{U}^{(n)}(t)^*)(\unit\otimes\cdot) - \varrho  \| <\eps.
\end{equation}
\end{definition}
Notice that this definition is stronger than \cite[Def.2.2]{ABFH} and probably more intuitive from a physical point of view.

Our goal is to construct an autonomous unitary dilation $(\B\otimes B(\H_0),  \omega_e, \hat{U}(t))$ such that dynamical decoupling works. Indeed we are going to constructively prove the following theorem:

\begin{theorem}\label{th:decoupling}
Given a finite-dimensional quantum system with an analytic quantum channel 
and a decoupling set as above. Then there is an autonomous unitary dilation such that dynamical decoupling works.
\end{theorem}

\begin{proof}
The proof is long, so let us first summarise the main ideas. We start by constructing a finite-dimensional unitary dilation with time-dependent Hamiltonian: more precisely, we consider the Choi map of $\Phi_t^*$ and use its eigenvalues and eigenvectors to construct an explicit form of Kraus operators for $\Phi_t^*$. Using complex analysis, linear algebra and analyticity properties of the Choi map, we derive a number of nice properties of this dilation, which in Step 2 allow us to dilate it to a larger infinite-dimensional space with constant Hamiltonian; the advantage of a constant Hamiltonian is that the framework of dynamical decoupling is likely to work. In fact, in the final step we compute the perturbed dynamics in the presence of dynamical decoupling operations and prove that dynamical decoupling does indeed work. Let us look at the details:

\emph{1. Let us construct a finite-dimensional unitary dilation with time-dependent Hamiltonian.}

This is a modification of the argument in~\cite{DMB}. The time evolution of a density matrix $\rho$ corresponding to the UCP time evolution $\Phi_t$ is CPTP and given by $\Phi_t^*$; it can be written as
\begin{equation}\label{eq:eps-Kraus}
\Phi_t^* (\rho) = \sum_{k=1}^{d_1} M_k(t) \rho M_k(t)^* \quad \forall t\ge 0,
\end{equation}
where $d_1\in\N$ and $M_k(t)\in B(\H_0)$ are certain Kraus operators. More precisely,
let $\Omega$ denote the maximally entangled unit vector in $\H_0\otimes\H_0$ given by $\Omega = \frac{1}{\sqrt{d_0}} \sum_{j=1}^{d_0} u_j \otimes u_j$, and consider the Choi map $t\mapsto (\Phi_t^* \otimes \id) |\Omega\rangle\langle\Omega | \in B(\H_0\otimes\H_0)$; it is positive and analytic for all $t\in [0,\infty)$, 
and hence its eigenvalues $\lambda_k(t)$, not identically 0, and eigenvectors $v_k(t)$, where $k=1,\ldots,d_1$ and $d_1\le d_0^2$, are analytic 
for $t\ge 0$, 
cf.~\cite[Thm.II.6.1]{Kato}. Moreover, positivity of the Choi map implies non-negativity of all $\lambda_k(t)$, namely $\lambda_k(t)\in[0,1]$ not identically 0, and $\sum_k \lambda_k(t)=1$. Then one defines
\[
M_k(t) 
= \sqrt{d_0} \sum_{i,j=1}^{d_0} \sqrt{\lambda_k(t)} \langle u_i\otimes u_j ,  v_k(t) \rangle |u_i\rangle\langle u_j| ,
\]
with $\{u_j\}$ an orthonormal basis of $\H_0$, and one can hence show that~\eqref{eq:eps-Kraus} holds.
Now the square-root is analytic on $\C\setminus (-\infty,0]$, hence $M_k(t)$ is analytic at $t\in (0,\infty)$ if $\lambda_k(t)>0$.

At $t=0$, since $\Phi_0=\id$, one has $(\Phi_0^* \otimes \id) |\Omega\rangle\langle\Omega |  = |\Omega\rangle\langle\Omega |$, whence one eigenvalue is 1, say $\lambda_1(0)=1$ (with $v_1(0)=\Omega$) and all the rest are 0, namely $\lambda_k(0)= 0$ for $k\geq 2$. Therefore, $\sqrt{\lambda_1(t)} = 1 + t \lambda'_1(0)/2 + O(t^2)$, with  $\lambda'_1(0)\leq0$, is analytic for $t\geq0$ in a neighbourhood of $0$.
On the other hand, for $k\geq2$, $\sqrt{\lambda_k(t)} = t^{1/2} \lambda'_k(0)^{1/2} (1 + O(t))$, with  $\lambda'_k(0)\geq0$, is continuous for $t\geq 0$ and analytic for $t>0$ in a neighbourhood of $0$, but in general is \emph{not} analytic at $0$, not even differentiable, its derivative diverging as $t^{-1/2}$. 

It is so, unless $\lambda'_k(0)=0$, in which case one can choose as square-root an analytic function at~$0$: $\sqrt{\lambda_k(t)} = t (\lambda''_k(0)/2)^{1/2} + O(t^2)$ with $\lambda''_k(0)\geq0$. This is always the case for positive times. Indeed, if for some $k\geq 1$ one has $\lambda_k(t_0)=0$ for $t_0>0$ then also $\lambda_k'(t_0)=0$ and one can choose as square root in the neighborhood of $t_0$ the analytic function $\sqrt{\lambda_k(t)} = (-1)^{q_k} (t_0-t) (\lambda''_k(0)/2)^{1/2} + O((t-t_0)^2)$, where $q_k$ is the number of zeroes of $\lambda_k$ in~$(0,t_0)$.

Hence we can henceforth assume $M_1(t)$ to be analytic for $t\geq0$, with $M_1(t) = \unit -t X + O(t^2)$ as $t\downarrow 0$, and $M_k(t)$, for $k\geq2$, to be continuous at all $t\ge 0$ and analytic at $t>0$, with $M_k(t) \sim t^{1/2} L_k$ as $t\downarrow 0$, and some $X,L_k\in B(\H_0)$.

Choosing $\H_1=\C^{d_1}$, $\{e_j\}$ the canonical basis, 
 and $\mu = \langle e_1, \cdot\, e_1 \rangle$  as an initial pure bath state, we define the time propagator $U(t,0)$ on $\H=\H_1\otimes\H_0$ by the equation
\begin{equation}
U(t,0) = \sum_{j,k=1}^{d_1} | e_k\rangle\langle e_j| \otimes M_{k,j}(t), 
\label{eq:U(t,0)}
\end{equation}
with
\[
\sum_{k=1}^{d_1} M_{k,j}(t)^* M_{k,\ell}(t) = \sum_{k=1}^{d_1} M_{j,k}(t) M_{\ell,k}(t)^* = \delta_{j,\ell} \unit,
 \quad  1\leq j, \ell \leq d_1.
\]
Setting $M_{k,1}(t) = M_k(t)$ for $1\leq k\leq d_1$  fixes $d_0$ columns of $U(t,0)$ (i.e. $d_0 d_1$-dimensional vectors in $\H$); indeed these $d_0$ columns turn out to be normalised and mutually orthogonal because $\Phi_t^*$ is trace-preserving. Moreover, by construction, these columns are continuous at $t\ge 0$ and analytic 
at $t>0$.

Next, we complete these $d_0$ vector fields to $d_0d_1$ vector fields forming an orthonormal basis of $\H$, at all $t\geq 0$. In practice, this can be achieved as follows. We start with $d_0 (d_1-1)$ constant vector fields such that together with the given $d_0$ ones they form a basis of $\H$ on a certain interval $[0,t_1]$. Then we choose another  $d_0 (d_1-1)$ constant vector fields such that together with the given $d_0$ ones they form a basis of $\H$ on a certain interval $[t_1,t_2]$, and so on. 
Now we make the transitions at $t_k$ smooth, 
this way obtaining $d_0d_1$ smooth 
vector fields such that at all $t$ they form a basis of $\H$. By continuity and compactness arguments one needs only  finitely many such $t_k$ between $0$ and any $t$. Finally, we can apply Gram-Schmidt in order to obtain a 
smooth orthonormal frame between $0$ and $t$, cf.~\cite[Sec.8]{Lee}. We then define the columns of $U(t,0)$ as this orthonormal frame, so $U(t,0)$ is completely defined (although the choice was not unique) and unitary and it is continuous at $t\ge 0$ and
smooth at $t>0$. 

We then obtain a unitary dilation $\bigl(B(\H_1)\otimes B(\H_0), \mu, U(t,0))$ of $(B(\H_0), \Phi_t\bigr)$ on the finite dimensional product space $\H=\H_1\otimes H_0$, namely
\begin{equation}
\label{eq:dil1}
\varrho \circ\Phi_t(A) = (\mu \otimes \varrho )\circ \Ad (U(t,0)^*)(\unit\otimes A),
\end{equation}
for all $A\in B(\H_0)$ and $t\ge 0$, and with the initial bath state $\mu=\langle e_1, \cdot e_1 \rangle$. The dilation is non-autonomous, since  $(\Ad U(t,0))_{t\ge 0}$ does not usually form a semigroup.

Since $U(t,0)$ is
smooth on $(0,\infty)$, we can define
\begin{equation}
H(t) = \rmi\frac{\rmd}{\rmd t} U(t,0) \, U(t,0)^*\quad \forall t\in(0,\infty),
\label{eq:H(t)}
\end{equation}
which must also be
smooth, while $H(t)$ may diverge like $H(t) = O(t^{-1/2})$ as $t\downarrow 0$, so that
\begin{equation}\label{eq:H-int}
\int_0^t \|H(x)\|\rmd x  = O(t^{1/2}), 
\quad t\downarrow 0.
\end{equation}

We extend the function $U(\cdot,0)$ from $[0,\infty)$ to all of $\R\times\R$, by setting
\[
U(t,0)=U(-t,0)^*, \quad \forall t\le 0,
\]
and
\[
U(0,t)=U(t,0)^* \quad \forall t\in\R
\]
and
\[
U(t,s)= U(t,0) U(0,s) \quad \forall s,t\in\R.
\]
This way, $U$ is well-defined. Moreover, since $M_k(t)$ are continuous for $t\ge 0$ and $\Phi_0=\id$, we get the continuity of $U(t,0)$, and also $U(0,0)=\unit$, and by extension continuity in both variables and hence a unitary propagator on $\H$ in the sense of~\cite[p.282]{Reed2}. Since $\H$ is finite-dimensional, the continuity is uniform on compact intervals. Apart from $\{0\}\times\R\cup \R\times\{0\}$, $U$ is actually smooth. 

\emph{2. We construct an autonomous unitary dilation.}

The idea is based on~\cite[p.291]{Reed2} and~\cite{How}. Namely, we consider an (infinite-dimensional) second bath $\K=L^2(\R)$ and $\hat{\H}=\K\otimes\H$, and define, for all $\xi\in \hat{\H} \simeq L^2 (\R , \H)$,
\begin{equation}\label{eq:Howland}
\big(\hat{U}(t) \xi\big) (x) = U(x,x-t) \xi(x-t) \quad \forall x,t\in\R.
\end{equation}
Then $\{\hat{U}(t)\}_{t\in\R}$ forms a strongly continuous one-parameter group of unitaries on $\hat{\H}$; in particular, it has an infinitesimal generator which may be interpreted as a constant though unbounded Hamiltonian. $\K$ is the Hilbert space of the clock from Section~\ref{sec:singular}. We choose here to follow Option~\ref{sub:1}, so our algebra is $\A=C_b(\R) \subset B(\K)$, and as initial state we consider $\nu_0$ (at $t=0$ the clock is in a point localised state at $x=0$). Then we have an autonomous unitary dilation 
$$
\bigl(\B\otimes B(\H_0), \omega_e,\hat{U}(t)\bigr) := \bigl(\A\otimes B(\H_1) \otimes B(\H_0), \nu_0 \otimes \mu, \hat{U}(t)\bigr)
$$ 
of $(B(\H_0),\Phi_t)$. 
To see that we have actually got a dilation, let us compute:
\begin{align*}
(\nu_0\otimes\mu\otimes\varrho ) & \big(\Ad (\hat{U}(t)^*)(\unit\otimes\unit\otimes A) \big) 
= (\nu_0\otimes\mu\otimes\varrho ) \big(\hat{U}(-t)(\unit\otimes\unit\otimes A)\hat{U}(t) \big)\\
=& (\mu\otimes\varrho ) \big(U(s,s+t)(\unit\otimes A)U(s+t,s+t-t) \big) \restriction_{s=0}\\
=& (\mu\otimes\varrho ) \big(U(0,t)(\unit\otimes A)U(t,0) \big) \restriction_{s=0}\\
=& \varrho ( \Phi_t(A)),
\end{align*}
for every $t\in [0,\infty)$, $A\in B(\H_0)$ and every state $\varrho $ on $B(\H_0)$, where the last line follows from~\eqref{eq:dil1}.

\emph{3. We compute the time evolution with dynamical decoupling and check whether it works.}

We let the ``averaged bath Hamiltonian'' be defined by
\[
H_e(s) :=\frac{1}{|V|}\sum_{v\in V} (\unit \otimes v) H(s) (\unit\otimes v^*) \in B(\H_1)\otimes\unit_{\H_0} \quad \forall s\in\R\setminus \{0\} 
\]
as in~\cite[Sec.1]{ABFH}, and denote the corresponding continuous time-evolution propagators by $U_e(t,s)$ on $\H$, and $\hat{U}_e(t)$ the corresponding one-parameter group on $\hat{\H}$ in analogy with~\eqref{eq:Howland}. Notice that $U_e(t,s)$ acts trivially on $B(\H_0)$ and in particular
\begin{equation}\label{eq:E1dilation}
(\mu\otimes\varrho )\circ \Ad (U_e(t,s)^*)(\unit\otimes\cdot) = \varrho .
\end{equation}

Given $t>0$ and $\eps>0$, we have to find a decoupling cycle $(v_1,\ldots,v_n)$ of a certain length $n\in |V|\N$ such that \eqref{eq:toprove} holds. We compute

\begin{equation*}
\begin{aligned}
\Big(\hat{U}^{(n)}(t)^* \xi\Big)(s) 
=& \Big(\hat{U}( \hat{v}_1\hat{U}(-\frac{t}{n}) \hat{v}_1^*\hat{v}_2\hat{U}(-\frac{t}{n})\hat{v}_2^* \cdots \hat{v}_n\hat{U}(-\frac{t}{n}) \hat{v}_n^* \xi\Big) (s)\\
=& \Big( \hat{v}_1 U(s,s+\frac{t}{n}) \hat{v}_1^* \cdots \hat{v}_n U(s+\frac{(n-1)t}{n},s+t) \hat{v}_n^*\Big) \xi(s+t) \\
=& U^{(n)}(s,s+t)\xi(s+t),
\end{aligned}
\end{equation*}
where
\[
U^{(n)}(s,s+t)= \hat{v}_1 U(s,s+\frac{t}{n}) \hat{v}_1^* \cdots \hat{v}_n U(s+\frac{(n-1)t}{n},s+t) \hat{v}_n^*.
\]
Thus
\begin{equation}\label{eq:claim-Un3}
\begin{aligned}
\| (\omega_e\otimes\varrho )\circ &  \Ad(\hat{U}^{(n)}(t)^*)(\unit\otimes\cdot) - \varrho  \|\\
=&\|(\nu_0\otimes\mu\otimes\varrho ) \circ\Ad\big(\hat{U}^{(n)}(t)^*\big)(\unit\otimes\unit\otimes \cdot)\\
&\quad - (\nu_0\otimes\mu\otimes\varrho ) \circ\Ad\big(\hat{U}_e(t)^*\big)(\unit\otimes\unit\otimes \cdot)
\|\\
=& \| (\mu\otimes\varrho )\circ \Ad\big( U^{(n)}(0, t)\big)(\unit\otimes\cdot) - (\mu\otimes\varrho )\circ \Ad\big( U_e(0, t)\big)(\unit\otimes\cdot)\|\\
\le& 2 \|U^{(n)}(0, t)- U_e(0, t)\|.
\end{aligned}
\end{equation}
Thus in order to complete the proof of \eqref{eq:toprove}, it remains to show that if we choose $n\in\N$ big enough then \eqref{eq:claim-Un3} is less than $\eps$, or in other words, that
\begin{equation}\label{eq:claim-Un4}
\| U^{(n)}(0,t)- U_e(0,t)\| < \frac{\eps}{2}.
\end{equation}

Since~\cite{ABFH} works in the framework of constant Hamiltonian whereas here $H(t)$ is not constant, we would like to relate the two frameworks. Choose $m\in\N$ and a piece-wise constant (step function) approximation $H_m$ of $H$ with $m$ equal length steps on $(0,t]$ such that
\begin{equation}\label{eq:eps6}
\int_{0}^{t} \| H(x)-H_m(x)\|\rmd x
<\frac{\eps}{6},\quad 
\int_{0}^{t} \| H_{m,e}(x)-H_{e}(x)\|\rmd x
<\frac{\eps}{6},
\end{equation}
where $H_{m,e}$ is the averaged bath Hamiltonian resulting from $H_m$ (in the same way as $H_e$ arises from $H$).
This is possible because of uniform continuity of $H$ and $H_e$ on compact subintervals of $(0,t]$ and because of \eqref{eq:H-int} applied to $H$ as well as to $H_e$. We denote the corresponding time evolution unitary families by $U_m$ and $U_{m,e}$.

Now for each of the $m$ steps we apply dynamical decoupling; more precisely, we apply the same decoupling cycles of length $q$ such that, with $n:= m q$, the resulting length of the total decoupling cycle, we have:
\[
\|U_m^{(n)}(0,t)-U_{m,e}(0,t)\|<\frac{\eps}{6}.
\]
This can always be achieved because dynamical decoupling always works on a finite-dimensional space $\H$ with constant Hamiltonian~\cite[proof of Thm.3.1]{ABFH}, and hence also for a piece-wise constant Hamiltonian with a fixed finite number of steps $m$.

Moreover, it follows from~\cite[Eq.(3)]{DMB} (see also~\cite{NDGD}) together with \eqref{eq:eps6} that
\begin{align*}
\| U^{(n)}(0,t)-U_m^{(n)}(0,t)\| \le& \int_{0}^{t} \| H(x)-H_m(x)\|\rmd x
<\frac{\eps}{6}
\end{align*}
and that
\[
\| U_{m,e}(0,t)-U_e(0,t)\| \le \int_{0}^{t} \| H_{m,e}(x)-H_e(x)\|\rmd x
<\frac{\eps}{6}.
\]

Putting these three inequalities together using an $\eps/3$-argument, we get
\begin{align*}
\|U^{(n)}&(t_0,t)-  U_{e}(0,t)\| 
\le  \|U^{(n)}(0,t)-U_m^{(n)}(0,t)\| \\
& + \|U_m^{(n)}(0,t)-U_{m,e}(0,t)\|
+  \|U_{m,e}(0,t)-U_{e}(0,t)\| \\
<& 3\frac{\eps}{6}=\frac{\eps}{2},
\end{align*}
proving~\eqref{eq:claim-Un4}, and hence verifying~\eqref{eq:toprove} for the present dilation.
\end{proof}

\begin{corollary}\label{cor}
Let $\mathcal{L}= \Phi'_0$, let $H(t)$ be the time-dependent Hamiltonian~\eqref{eq:H(t)} of the first unitary dilation,  and let $E_n(\varrho)$ be  the decoupling error~\eqref{eq:toprove} for the autonomous unitary dilation constructed above. Then  $H(t)$ is uniformly bounded and smooth iff $\mathcal{L} = \rmi [K,\cdot]$ for some $K=K^*$. In such a case $E_n(\varrho) = O(n^{-1})$ as $n\to\infty$. Otherwise, $H(t)$ is unbounded at $0$, $H(t) = t^{-1/2} H_0 + O(t)$ as $t\downarrow 0$, for some $H_0=H_0^*$, and $E_n(\varrho) = O(n^{-1/2})$ as $n\to\infty$.
\end{corollary}
\begin{proof}
In Step 1 of the proof of Theorem~\ref{th:decoupling}, we have shown that $M_1(t)$ is analytic for $t\geq0$, with $M_1(t) = \unit -t X + O(t^2)$ as $t\downarrow 0$, and $M_k(t)$, for $k\geq2$, is continuous at all $t\ge 0$ and analytic at $t>0$, with $M_k(t) = t^{1/2} L_k + O(t)$ as $t\downarrow 0$. 
By plugging them into the Kraus decomposition~\eqref{eq:eps-Kraus} one has
\[
\Phi_t^*(\rho) = \rho -t (X\rho + \rho X^*) + t \sum_{k=2}^{d_1} L_k \rho L_k^* + O(t^2),
\]
and thus the derivative of the time evolution at $t=0$ reads ${\Phi_0^*}'=\mathcal{L^*}$, with
\[
\mathcal{L^*}(\rho) = -\rmi [K,\rho] -\frac{1}{2} \sum_{k=2}^{d_1} \bigl(  \{L_k^* L_k, \rho \} -2 L_k\rho L_k^*\bigr)
\]
having the GKLS form~\cite{GKS,Lindblad}. Here $K= \operatorname{Im} X = (X-X^*)/(2\rmi)$ and $\sum_{k=2}^{d_1} L_k^* L_k = \operatorname{Re} X = (X+X^*)/2$ by trace preservation. 

The operators $M_k(t)$, for $k\geq2$, are analytic at $0$ iff $L_k=0$ for all $k\geq 2$, that is iff $\mathcal{L^*}(\rho) = -\rmi [K,\rho]$. In such a case the unitary propagator $U(t,0)$ in~\eqref{eq:U(t,0)} is smooth for $t\geq0$ and uniformly continuous, and so is the Hamiltonian $H(t) = \rmi\frac{\rmd}{\rmd t} U(0,t) \, U(t,0)^*$.  For fixed $t\ge 0$, set  $t_k=\frac{k t}{m}$ and choose the stepwise approximation $H_m$ such that $H_m(s) = H(t_k)$ for $s\in [t_{k-1}, t_k)$, with $k=1,\dots,m$. We therefore get
\[
\int_{0}^{t} \| H(x)-H_m(x)\|\rmd x = O\Bigl(\frac{1}{m}\Bigr),\quad 
\int_{0}^{t} \| H_{m,e}(x)-H_{e}(x)\|\rmd x  = O\Bigl(\frac{1}{m}\Bigr).
\]
By using the second iteration of the equation for the propagator $U_m^{(q)}$ 
\begin{align*}
U_m^{(q)}(t_{k-1},t_k) = & \unit +\rmi \int_{t_{k-1}}^{t_k}  U_m^{(q)}(t_{k-1},s) H_m^{(q)}(s) \rmd s
\\
= & \unit +\rmi \int_{t_{k-1}}^{t_k} H_m^{(q)}(s) \rmd s - \int_{t_{k-1}}^{t_k}   \int_{t_{k-1}}^s U_m^{(q)}(t_{k-1},r)  H_m^{(q)}(r)  H_m^{(q)}(s) \rmd r \rmd s
\end{align*}
and analogously for $U_{m,e}$, one has
\begin{equation}\label{eq:Um-estimate}
\begin{aligned}
\|U_m^{(n)}(0,t)-  U_{m,e}(0,t)\| \le & \sum_{k=1}^m \|U_m^{(q)}(t_{k-1},t_k)-U_{m,e}(t_{k-1},t_k)\| \\
= &  \sum_{k=1}^m \Big\| \int_{t_{k-1}}^{t_k} \!\! \int_{t_{k-1}}^s \Bigl[ U_m^{(q)}(t_{k-1},r) 
H_m^{(q)}(r)  H_m^{(q)}(s) 
\\
& \qquad\qquad\qquad\qquad\quad - U_{m,e}(t_{k-1},r)  H_{m,e}(r)  H_{m,e}(s)
\Bigr] \rmd r \rmd s \Big\| 
\\
\le &  \sum_{k=1}^m  \int_{t_{k-1}}^{t_k} \int_{t_{k-1}}^s \bigl(
\| H_m^{(q)}(r)\|  \|H_m^{(q)}(s)\| + \|H_{m,e}(r)\|  \|H_{m,e}(s)\|
\bigr) \rmd r \rmd s
\\
\le &   \frac{t^2}{m^2} \sum_{k=1}^m   \| H(t_k) \|^2.
\end{aligned}
\end{equation}
In the second line we used the decoupling condition, namely $\frac{1}{|V|} \sum_{v\in V} \hat{v} H_m(s) \hat{v}^* = H_{m,e}(s)$, which implies that the terms of order 0 and 1 cancel with each other, hence the equality. In the fourth line we used $\|H_m^{(q)}(s)\| = \| H(t_k)\|$ and $\|H_{m,e}(s)\| = \| H_{e}(t_k)\| \leq \| H(t_k)\|$, for $s\in(t_{k-1}, t_k)$. Finally, by setting $\eta(t)= \max_{s\in[0,t]}\|H(s)\| <\infty$, one has
\[
\|U_m^{(n)}(0,t)-  U_{m,e}(0,t)\| \le  \frac{\eta^2 t^2}{m} = O\Bigl(\frac{1}{m}\Bigr),
\]
as $m\to\infty$. Since $n= m q$, by keeping the cycle length $q$ fixed, one gets using the $\eps/3$-argument as in Step 3 of the proof of Theorem \ref{th:decoupling} that $E_n(\rho) = O(1/m)= O(1/n)$ as $n\to\infty$.

On the other hand, we saw that if $L_k\neq 0$ for some $k\geq 2$ then $U(t,0) = \unit -\rmi t^{1/2} H_0 +O(t)$ and $H(t) = t^{-1/2} H_0 + O(t)$ as $t\downarrow0$, so that~\eqref{eq:H-int} holds. This, together with uniform continuity on compact subintervals of $(0,t]$, implies
\[
\int_{0}^{t} \| H(x)-H_m(x)\|\rmd x = O\Bigl( \sqrt{\frac{1}{m}}\Bigr)
,\quad 
\int_{0}^{t} \| H_{m,e}(x)-H_{e}(x)\|\rmd x  = O\Bigl( \sqrt{\frac{1}{m}}\Bigr),
\]
for fixed $t\ge 0$ as $m\ra\infty$. 
Moreover, by setting $\tilde{\eta}(t)= \max_{s\in[0,t]}s \|H(s)\|^2<\infty$, one has from~\eqref{eq:Um-estimate} 
\begin{align*}
\|U_m^{(n)}(0,t)-U_{m,e}(0,t)\|  \le   \frac{t^2}{m^2} \sum_{k=1}^m   \frac{\tilde{\eta}}{t_k} = \frac{\tilde{\eta} t}{m} \sum_{k=1}^m   \frac{1}{k} 
\leq  \frac{\tilde{\eta} t }{m}  (1+\log m) =  O\Bigl( \frac{\log m}{m}\Bigr),
\end{align*}

as $m\to\infty$. Therefore, the $\eps/3$-argument now yields $E_n(\varrho) =  O(m^{-1/2})= O(n^{-1/2})$ as $n=mq \to\infty$ with fixed $q$.
\end{proof}

An immediate consequence of Corollary \ref{cor} is that for a quantum dynamical semigroup, $\Phi_t = \rme^{t \mathcal{L}}$, the time-dependent Hamiltonian is singular at $t=0$ and the decoupling error is of order $n^{-1/2}$, unless the evolution is unitary, $\Phi_t = \Ad (\rme^{\rmi t K})$, in which case the decoupling error is of order $n^{-1}$. Indeed, in the numerical simulation of the single-qubit amplitude damping channel shown in Fig.~\ref{simulation}  the decoupling error decays  as $(t/\Delta t)^{-1/2}$.

\begin{corollary}\label{cor2}
Every quantum dynamical semigroup on a finite-dimensional quantum system has got an autonomous unitary dilation for which dynamical decoupling works and another autonomous unitary dilation for which dynamical decoupling does not work.
\end{corollary}

\begin{proof}
The first part is Theorem~\ref{th:decoupling}, the second part is the result in~\cite{GN}.
\end{proof}

\begin{remark} \label{remarkeuler}
\begin{itemize}
\item[(i)] The autonomous unitary dilation in the proof of Theorem \ref{th:decoupling} is actually strongly continuous. We decided for Option \ref{sub:1} for the sake of simplicity but Options \ref{sub:1a}, \ref{sub:2} or their extensions would work equally well. We could also use Option \ref{sub:3} but in that case the Hilbert space $\K$ would become nonseparable and strong continuity would be lost. Actually our choice shows that the clock can always be chosen to be \emph{classical}.
\item[(ii)] We worked exclusively in the context of bang-bang dynamical decoupling, i.e. with infinitely strong pulses. This is the standard framework of dynamical decoupling and the most elegant one to deal with mathematically although it is quite idealised and arguably unphysical. However, our construction should also work for most decoupling procedures with bounded decoupling operations given that \eqref{eq:Howland} reduces essentially everything to finite dimensions. Concretely, going through the proof of Theorem \ref{th:decoupling} one can immediately verify that Euler dynamical decoupling (cf.~\cite{VK03} and also~\cite[Sec.3.2]{ABH17}) works.
\item[(iii)] A natural question arising from Corollary~\ref{cor2} is: which dilation (if any) is the physically correct one? Is there a way of finding out? Or can dynamical decoupling actually be used to determine which dilation is the correct one, i.e., which environment is the actual physical environment? That would be an interesting variation of~\cite{AHFB15}. See also \cite{Wiseman}.
\item[(iv)] Another natural question of rather mathematical interest is: is there a more general method to dilate CPTP dynamics beyond the analytic case?
\end{itemize}
\end{remark}

\section{Conclusion}

We have developed a universal method to dilate a given analytical CPTP dynamics on a finite-dimensional quantum system to a time-independent Hamiltonian evolution on a larger space. This includes the important special case of Markovian (i.e. GKLS-generated) dynamics, for which the Gregoratti-Chebotarev dilation was developed. Common properties between our dilation and the Gregoratti-Chebotarev dilation are that, at least in the case where the original Markovian dynamics has population decay, both Hamiltonians are unbounded below and above as this is a necessary condition for any dilation of a Lindblad dynamics with population decay~\cite[eq.(21)]{JEK}.
While our dilation has a delta singularity in time in the state, Gregoratti-Chebotarev has a delta singularity in space in the Hamiltonian~\cite{GoughG}.
This duality is reflected in the response of the dilations to interventions, with Gregoratti-Chebotarev resistant to dynamical decoupling, whereas our dilation can be decoupled easily. This shows the strong influence of the choice of dilation for a given control situation. While in an ideal world, information of the actual environment determines such a choice, in practice often only reduced dynamics is known, meaning that a variety of environment dynamics could be considered.
In quantum optics, Gregoratti-Chebotarev has often been a natural choice, because this dilation arises naturally from the physical setup~\cite{Gardiner}, but in solid state systems this is less straight forward. As explained in the introduction, the fact that dynamical decoupling has been successfully used in a number of experimental realisations \cite{NISQ22} shows that Gregoratti-Chebotarev might not be a good choice in those settings. In this article we have taken the first steps into the broader world of dilations beyond Gregoratti-Chebotarev by providing a new alternative. Our dilation provides a mathematically easier setup and much more flexibility since it can describe any analytic (not necessarily Markovian) reduced dynamics. Moreover, the existence of a specific physical model \cite{Burgarth19} which shares features of both dilations shows the rich mathematical and physical structure of quantum noise. The above observation about the singularity is a good qualitative indication as to whether dynamical decoupling will work in a given model though a proper classification or a rigorous and easily applicable criterion is out of reach at the moment.

\subsection*{Acknowledgments}
DB acknowledges support by the Australian Research Council (project number FT190100106).
PF is partially supported by INFN through the project `QUANTUM', by the Italian National Group of Mathematical Physics (GNFM-INdAM), and by Regione Puglia and QuantERA ERA-NET Cofund in Quantum Technologies (GA No. 731473), project PACE-IN.
RH would like to thank Veronique Fischer for explanations about defect measures and other topics in harmonic analysis relevant to Option~\ref{sub:1a} in Section~\ref{sec:singular} and for pointing out to him reference~\cite{Cor}. 

\subsection*{Note added in proof} After the manuscript had been accepted, J. Kiukas made us aware of \cite{KSW}, where a construction is proposed that has a similar underlying idea as the one presented in our Section \ref{sec:singular}. While the two constructions are independent and different, the paper still ought to be cited.

\end{document}